\documentclass[11pt]{article}

\usepackage[T1]{fontenc}
\usepackage{lmodern}
\usepackage[letterpaper,margin=1in]{geometry}
\usepackage{microtype}
\usepackage{amsmath,amssymb,amsthm,mathtools}
\usepackage{booktabs}
\usepackage{longtable}
\usepackage{tabularx}
\usepackage{array}
\usepackage{enumitem}
\usepackage{xcolor}
\usepackage{listings}
\usepackage{float}
\usepackage{hyperref}
\usepackage[capitalise,noabbrev]{cleveref}

\hypersetup{
  colorlinks=true,
  linkcolor=blue!45!black,
  citecolor=green!35!black,
  urlcolor=blue!55!black,
  pdftitle={Machine-Checked Arithmetic Bit Complexity of the Kannan--Bachem Smith Normal Form in Lean 4},
  pdfauthor={Junye Ji}
}

\setlist{nosep}
\newtheorem{theorem}{Theorem}[section]
\newtheorem{lemma}[theorem]{Lemma}
\newtheorem{proposition}[theorem]{Proposition}

\theoremstyle{definition}

\AddToHook{env/theorem/begin}{\crefalias{section}{theorem}}
\AddToHook{env/lemma/begin}{\crefalias{section}{lemma}}
\AddToHook{env/proposition/begin}{\crefalias{section}{proposition}}
\AddToHook{env/corollary/begin}{\crefalias{section}{corollary}}
\AddToHook{env/definition/begin}{\crefalias{section}{definition}}
\AddToHook{env/remark/begin}{\crefalias{section}{remark}}

\newcommand{\Z}{\mathbb{Z}}
\newcommand{\N}{\mathbb{N}}

\newcommand{\bitlen}{\operatorname{bitlen}}
\newcommand{\natsize}{\operatorname{size}_2}
\newcommand{\enc}{\operatorname{enc}}
\newcommand{\cost}{\operatorname{cost}}
\newcommand{\AddCost}{\operatorname{AddCost}}
\newcommand{\MulCost}{\operatorname{MulCost}}
\newcommand{\DivCost}{\operatorname{DivCost}}

\newcommand{\Lean}{\textsc{Lean}}
\newcommand{\mathlib}{\textsc{mathlib}}
\newcommand{\code}[1]{\texttt{\detokenize{#1}}}
\newcommand{\decl}[1]{\nolinkurl{#1}}
\newcommand{\mcode}[1]{\text{\ttfamily\detokenize{#1}}}
\newcommand{\concat}{\mathbin{\|\!\|}}
\definecolor{codebg}{RGB}{247,247,247}
\definecolor{codekw}{RGB}{28,73,122}
\definecolor{codecom}{RGB}{85,105,85}
\lstdefinelanguage{Lean}{
  morekeywords={def,theorem,structure,where,let,by,exact,public,if,then,else,match,with,fun,forall,exists},
  sensitive=true,
  morecomment=[l]{--},
  morecomment=[s]{/-}{-/},
  morestring=[b]"
}
\title{Machine-Checked Arithmetic Bit Complexity of the\\
Kannan--Bachem Smith Normal Form in \Lean~4}
\author{Junye Ji\\University of Washington}
\date{July 2026}

\begin{document}
\maketitle

\begin{abstract}
We formalize in \Lean~4 the Kannan--Bachem Smith normal form algorithm for nonsingular square integer matrices.  The program returns $S,U,U^{-1},V,V^{-1}$ and proves $UAV=S$, $U^{-1}SV^{-1}=A$, four inverse identities, the Smith divisibility conditions, and equality of $S$ with a canonical reference matrix.  Stabilization terminates because each recursive pass strictly decreases the binary size of the active pivot; the outer algorithm recurses on the lower-right block.  The computation also emits a flat trace of designated sign-magnitude arithmetic leaves.  Branch conditions, quotients, Bezout data, and matrix entries are taken from the recorded primitive runs.  Composite phases form their traces by concatenating the charge lists returned by the executed children.  Verified self-delimiting codecs define the input and output sizes.  Coefficient and work recurrences, closed by a kernel-checked polynomial-envelope calculus, give fixed polynomial bounds for both trace cost and the encoded length of the five output matrices.  The theorem concerns these arithmetic primitives; structural operations and compiled Lean runtime are outside the model.
\end{abstract}

\section{Introduction}

For $A\in\Z^{n\times n}$, Smith normal form gives a diagonal representative of the two-sided unimodular equivalence class of $A$.  In the nonsingular square case considered here, the target is
\[
  S=\operatorname{diag}(s_1,\ldots,s_n),
  \qquad 0<s_1\mid s_2\mid\cdots\mid s_n,
\]
with unimodular $U,V$ such that $UAV=S$.  The diagonal matrix is canonical, but the transforming matrices are not.  Returning $U,V$ and their inverses makes the equivalence checkable in both directions.

The algorithmic issue is intermediate coefficient growth.  Standard gcd-based elimination can create large integers even though the normal form exists.  Kannan and Bachem gave Hermite and Smith algorithms whose algebraic operation counts and intermediate binary lengths are polynomially bounded in the binary input length; they also construct the multiplier matrices~\cite{kannan1979}.  Their Smith procedure alternates left-Hermite and Hermite phases.  When a divisibility test fails, the next pass replaces the active pivot by a proper divisor.  That step controls both termination and coefficient growth.

We formalize this algorithm in \Lean~4~\cite{demoura2021} over \mathlib~\cite{mathlib2020}.  The input is a nonsingular square integer matrix.  The resource theorem counts specified binary sign-magnitude arithmetic primitives; \cref{sec:costboundary} states the boundary of the model.

We prove five main results.

\begin{enumerate}[label=(\arabic*)]
  \item The executable result contains $(S,U,U^{-1},V,V^{-1})$.  It proves the forward and reverse transformation equations, the left and right inverse identities for both transformations, the Smith predicate, and equality of $S$ with a canonical reference construction.

  \item Totality uses exactly two descents.  Stabilization is well founded on \code{Nat.size} of the active pivot's absolute value, and the outer recursion decreases the matrix dimension.

  \item Principal HNF, bounded-column HNF, Smith searches, clearing, injection, and four-product certificate composition are value-producing instrumented executions.  Their arithmetic leaves form the trace whose cost is folded by the final theorem.

  \item A prefix-free codec is verified for square matrices and for the five-matrix output.  The final size parameter is the codec length, which includes the dimension and every matrix entry.

  \item A small polynomial-envelope calculus closes the coefficient and work recurrences.  The transparent degree definitions reduce to $2{,}150{,}687$ for arithmetic cost and $98{,}990$ for output encoding length.  These are conservative witnesses produced by the generic closure rules, not sharp exponents.
\end{enumerate}

Coq and Isabelle developments prove Smith normal form in broader algebraic settings~\cite{cano2016,divason2022}.  Here the scope is integer square matrices so that one executable Kannan--Bachem reduction, its inverse certificates, and its arithmetic cost can be verified together.  The closest methodological precedent is the efficient integer implementation and polynomial-time analysis of LLL in Isabelle/HOL~\cite{bottesch2018}.

\section{Specification and concrete encodings}
\label{sec:statement}

\subsection{The five-matrix result}

For a square input $A$, the result contains five data matrices
\[
   S,\quad U,\quad U^{-1},\quad V,\quad V^{-1}.
\]
The reduction accumulates the inverse matrices, so no determinant or adjugate reconstruction is needed after computing the Smith matrix.  The final record proves
\begin{align}
  UAV &= S, \label{eq:forward}\\
  U^{-1}SV^{-1} &= A, \label{eq:backward}\\
  U^{-1}U=UU^{-1}&=I, \label{eq:uinv}\\
  V^{-1}V=VV^{-1}&=I. \label{eq:vinv}
\end{align}
It also proves the finite-index Smith predicate and identifies $S$ with the canonical reference Smith matrix already specified in the library.  Only $S$ is canonical.  The output-size theorem concerns the particular deterministic certificate produced by this implementation.

\subsection{Binary size and integer atoms}

For $m\in\N$, write
\[
 \natsize(m)=\mcode{Nat.size}(m)=
 \begin{cases}
  0,&m=0,\\
  1+\lfloor\log_2 m\rfloor,&m>0.
 \end{cases}
\]
For $z\in\Z$, define $\bitlen(z)=\natsize(|z|)$.  This is the implementation's \code{integerBitLength}.

A raw payload is made self-delimiting by pair framing:
\[
  0\mapsto 00,\qquad 1\mapsto 11,\qquad
  \text{terminator}\mapsto 01,
\]
with $10$ rejected.  The canonical integer payload is one sign/zero bit followed, for a nonzero value, by the positive binary magnitude.  Framing gives the exact atom length
\begin{equation}
  |\enc_{\Z}(z)|=2\bigl(\bitlen(z)+2\bigr).
  \label{eq:intatom}
\end{equation}
The decoder rejects the negative-zero payload.

\subsection{Matrix and Smith-output codecs}
\label{sec:codec}

The output codec is assembled from a typed matrix codec.  For a packed square matrix $M\in\Z^{n\times n}$,
\begin{align*}
  \mcode{matrixPayload}(M)
   &= \mcode{frame}(\mcode{encodeNat}(n))
      \concat \mathop{\concat}_{i=1}^{n}
      \mathop{\concat}_{j=1}^{n} \enc_{\Z}(M_{ij}),\\
  \enc_{\mathrm{mat}}(M)
   &= \mcode{frame}(\mcode{matrixPayload}(M)).
\end{align*}
Thus every encoded matrix carries its own framed dimension atom and its own outer frame.  If
\[
  \mcode{matrixBinarySize}(M)=|\mcode{matrixPayload}(M)|,
\]
then
\begin{equation}
  |\enc_{\mathrm{mat}}(M)|
    =2\,\mcode{matrixBinarySize}(M)+2.
  \label{eq:matrixlen}
\end{equation}

The Smith output is an outer frame around five complete typed-matrix encodings:
\begin{equation}
\begin{split}
  \enc_{\mathrm{out}}(r)=\mcode{frame}(&
    \enc_{\mathrm{mat}}(S)\concat
    \enc_{\mathrm{mat}}(U)\concat
    \enc_{\mathrm{mat}}(U^{-1})\concat{}\\[-0.2em]
    &\enc_{\mathrm{mat}}(V)\concat
    \enc_{\mathrm{mat}}(V^{-1})).
\end{split}
\label{eq:outputcodec}
\end{equation}
Consequently,
\begin{equation}
\begin{split}
 |\enc_{\mathrm{out}}(r)|=2\bigl(&|\enc_{\mathrm{mat}}(S)|
  +|\enc_{\mathrm{mat}}(U)|
  +|\enc_{\mathrm{mat}}(U^{-1})|\\[-0.2em]
  &+|\enc_{\mathrm{mat}}(V)|
  +|\enc_{\mathrm{mat}}(V^{-1})|\bigr)+2.
\end{split}
 \label{eq:outputlenexpanded}
\end{equation}

Both codecs use prefix decoders and prove a streaming round trip with an arbitrary suffix:
\[
 \operatorname{decodePrefix}(\enc(x)\concat t)
   =\operatorname{some}(x,t).
\]
This implies prefix-freeness.  The output decoder returns a data-only packed record containing the five matrices; proof fields are not serialized.  Encoding length counts symbols in an abstract binary alphabet, not bytes in Lean's heap representation of \code{List Bool}.

For later closed bounds, if every entry of an $n\times n$ matrix has bit length at most $b$, the verified matrix bound is
\begin{equation}
  L_{\mathrm{mat}}(n,b)
   =2\left(2\natsize(n)+2+n^2\,2(b+2)\right)+2,
 \label{eq:matrixencodingbound}
\end{equation}
and the five-matrix output bound is
\begin{equation}
  L_{\mathrm{out}}(n,b)=2\bigl(5L_{\mathrm{mat}}(n,b)\bigr)+2.
 \label{eq:outputencodingbound}
\end{equation}

\subsection{Main theorem and explicit witnesses}

Let $\mathcal E(A)$ be the instrumented execution defined in \cref{sec:execution}, and let $T(A)$ be its flat arithmetic trace.  Define
\[
  \cost(A)=\sum_{e\in T(A)}\cost(e).
\]
The source defines
\begin{align*}
 d_c&=\mcode{smithCostPolynomialDegree}, &
 C_c&=\mcode{smithCostPolynomialCoefficient}=2^{d_c},\\
 d_o&=\mcode{smithOutputPolynomialDegree}, &
 C_o&=\mcode{smithOutputPolynomialCoefficient}=2^{d_o}.
\end{align*}

\begin{theorem}[Machine-checked endpoint]
\label{thm:main}
For every $n$ and every $A\in\Z^{n\times n}$ with $\det A\ne0$, the execution returns $S,U,U^{-1},V,V^{-1}$ satisfying \eqref{eq:forward}--\eqref{eq:vinv}.  The matrix $S$ is in Smith normal form and equals the canonical reference Smith matrix.  Moreover,
\begin{align}
  \cost(A)
   &\le C_c\bigl(|\enc_{\mathrm{mat}}(A)|+1\bigr)^{d_c},
      \label{eq:maincost}\\
  |\enc_{\mathrm{out}}(S,U,U^{-1},V,V^{-1})|
   &\le C_o\bigl(|\enc_{\mathrm{mat}}(A)|+1\bigr)^{d_o}.
      \label{eq:mainoutput}
\end{align}
The four constants are independent of $n$ and of the entries of $A$.
\end{theorem}

In the accompanying Lean artifact, definitional reduction of the two transparent degree fields gives
\begin{equation}
 d_c=2{,}150{,}687,
 \qquad
 d_o=98{,}990.
 \label{eq:reduceddegrees}
\end{equation}
Thus $C_c=2^{2{,}150{,}687}$ and $C_o=2^{98{,}990}$.  Appendix~\ref{app:ledger} displays the recurrence and degree ledger from which these reductions arise.

The formal endpoint \decl{verifiedSmithPolynomialBitCost} returns the execution and the proofs summarized in \cref{thm:main}.

\section{The Kannan--Bachem computation}
\label{sec:algorithm}

\subsection{Correspondence with the 1979 algorithm}

The program follows the control structure of Kannan and Bachem's Smith procedure.  Several operations are transposed, and every transformation carries an explicit inverse.  \Cref{tab:kbmap} records the correspondence.

\begin{table}[ht]
\centering
\footnotesize
\begin{tabularx}{\linewidth}{@{}p{0.16\linewidth}p{0.30\linewidth}X@{}}
\toprule
Kannan--Bachem & Paper-level operation & Formal mechanism \\
\midrule
HNF Step 2 & column permutation in the original orientation; row preparation after transposition & deterministic complementary-minor scan \\
SNF Step 4 & left HNF of the active column & bounded-column / left-Hermite execution \\
SNF Step 5 & HNF of the active lower block & prepared principal HNF on a transpose \\
SNF Step 6 & repeat if the active column is not cleared & costed below-pivot search, another pass \\
SNF Step 7 & find an entry not divisible by the pivot and add its column to the active column & exact clear, costed lower-block search, injection, another pass \\
\bottomrule
\end{tabularx}
\caption{Correspondence with the original HNF and SNF steps~\protect\cite{kannan1979}.}
\label{tab:kbmap}
\end{table}

The formalization adds the mechanisms in \cref{tab:changes}; none changes the underlying Smith reduction.

\begin{table}[ht]
\centering
\footnotesize
\begin{tabularx}{\linewidth}{@{}p{0.31\linewidth}p{0.20\linewidth}X@{}}
\toprule
Mechanism & Status & Reason \\
\midrule
Deterministic complementary-minor preparation & equivalent implementation choice & supplies the principal-minor precondition without an external row-reduction oracle \\
Bird division-free determinant & arithmetic backend & gives executable determinant values and scalar add/multiply traces without division side conditions \\
Explicit $U^{-1}$ and $V^{-1}$ & strengthened certificate & supports the reverse equation and local inverse checking \\
Value-producing arithmetic trace & proof instrumentation & aligns branch data and matrix entries with the charged primitive runs \\
Four dense products per certificate composition & explicit materialization & pays for forward and inverse products on both sides \\
\bottomrule
\end{tabularx}
\caption{Implementation mechanisms not stated as such in the 1979 pseudocode.}
\label{tab:changes}
\end{table}

\subsection{Prepared principal Hermite reduction}

The principal HNF coefficient analysis assumes that every leading principal minor is nonzero.  An arbitrary nonsingular matrix need not have that property.  At dimension $m+1$, preparation evaluates the $m+1$ complementary minors obtained by deleting one row and the last column, selects the first nonzero minor, moves the corresponding row last, and recurses on the selected minor.  Expansion along the last column proves that some candidate is nonzero.

Determinants are evaluated by a verified implementation of Bird's division-free algorithm~\cite{bird2011}.  The scan binds each costed determinant run once and branches on its value; there is no uncharged fallback scan.

After preparation, the principal HNF kernel has three scalar transition forms.

\begin{description}[style=nextline,leftmargin=1.5em]
  \item[Bezout elimination.] One opaque bounded-Bezout-block run returns the gcd, two Bezout coefficients, two exact quotients when the gcd is nonzero, and mutually inverse $2\times2$ row matrices.

  \item[Reduction above a pivot.] One quotient/remainder run determines the row-addition coefficient.

  \item[Normalization.] One unit run returns the canonical unit $\pm1$.
\end{description}

A transition state contains a current matrix $B$, a multiplier $U$, and an explicit inverse $U^{-1}$.  If the scalar transition returns $F$ and $F^{-1}$, three instrumented dense products compute
\[
  B'=FB,\qquad U'=FU,\qquad U'^{-1}=U^{-1}F^{-1}.
\]
These products preserve $UA=B$ and both inverse identities.  Their traces also account for the arithmetic used to materialize all three updated matrices.

\subsection{One Smith pass}

A pass on a positive-dimensional active block has two phases.

\begin{enumerate}
  \item The bounded-column HNF execution applies row operations until the first column is in one-column Hermite form.
  \item Prepared principal HNF is applied to the transpose of the result.  Transposing its left transformation gives a right transformation for the original orientation.
\end{enumerate}

The two phase certificates are composed.  Composition executes four dense products: forward and inverse products on the left and on the right.  The pass theorem gives a nonzero normalized pivot and a cleared first row; the new pivot divides every entry of the pass input's first column.

\subsection{Pivot stabilization: entry contract and postcondition}
\label{sec:stabilize}

The stabilization recursion is entered with the contract produced by a pass:
\begin{enumerate}
  \item $p=A_{00}\ne0$;
  \item $p$ is normalized;
  \item the first row is zero away from $p$.
\end{enumerate}
The first column need not yet be cleared, and $p$ need not yet divide the lower-right block.  Those are stable-pivot \emph{postconditions}, not loop-entry invariants.

The control flow is:

\begin{center}
\begin{minipage}{0.94\linewidth}
\begin{lstlisting}[language={},basicstyle=\ttfamily\small]
stabilizeFrom(A):
  search below p for A[i,0] with p not dividing A[i,0]
  if found:
      P := passExecution(A)
      return compose(P, stabilizeFrom(P.result))

  C := clear all divisible entries below p by exact shears
  search C's lower-right block for C[i,j] with p not dividing C[i,j]
  if none exists:
      return C                 -- stable-pivot postcondition holds

  I := inject C[i,j] into the first column without changing p
  P := passExecution(I.result)
  return compose(C, I, P, stabilizeFrom(P.result))
\end{lstlisting}
\end{minipage}
\end{center}

Because the entry contract does not clear the first column, the function searches that column first.  If no witness is found, every entry below $p$ is divisible by $p$, and exact shears clear the column.  If the lower search also finds no witness, the result satisfies the full stable-pivot predicate:
\begin{enumerate}
  \item $p$ is nonzero and normalized;
  \item the first row and first column are cleared;
  \item $p$ divides every lower-right entry.
\end{enumerate}
If the lower search succeeds, injection preserves the pivot and moves the selected witness into the first column.  The next pass therefore falls under the same proper-divisor descent as the direct first-column branch.

\subsection{Outer recursion}

A stable matrix has block form
\[
  \begin{pmatrix}p&0\\0&B\end{pmatrix},
  \qquad p\mid B_{ij}.
\]
The lower-right block is nonsingular because the stabilized matrix is two-sided equivalent to the input and
\[
  \det\begin{pmatrix}p&0\\0&B\end{pmatrix}=p\det B.
\]
The algorithm recursively computes the Smith form of $B$, lifts that certificate by adjoining $p$, and composes it with the stabilization certificate.  The dimension decreases by one.

\section{Correctness and canonicality}
\label{sec:correctness}

\subsection{Local reversible transformations}

A reversible row step stores a forward matrix $F$, a backward matrix $F^{-1}$, and proofs
\[
  F^{-1}F=I,\qquad FF^{-1}=I.
\]
Column steps are obtained by transposition.  Permutations, row additions, unit scalings, bounded Bezout blocks, clearing shears, and witness injections instantiate this interface.  Composition assembles inverse matrices in reverse order.  The stored inverse identities then prove unimodularity without reconstructing it from determinants.

\subsection{Pass and stabilization facts}

The two Hermite phases imply that a pass returns a nonzero normalized pivot, a cleared first row, and divisibility of every old first-column entry by the new pivot.  Clearing preserves the pivot and first row and zeros the first column whenever the required divisibilities hold.  Injection preserves the pivot and copies a selected lower-right witness into the first column.  If both searches fail after clearing, the stable-pivot predicate follows directly; no generic Smith fallback is invoked.

\subsection{Assembly and canonicality}

After recursive Smith reduction of $B$, the lifted lower certificate acts only on the lower block.  Since $p$ divides every entry of $B$, it divides every entry after multiplication by integer matrices, hence divides the first invariant factor of the recursive Smith form.  Adjoining $p$ therefore yields a Smith divisibility chain for the full matrix.

Induction on dimension establishes the forward equation, reverse equation, inverse identities, and Smith predicate.  The result is then equated to the semantic reference matrix by uniqueness of normalized Smith forms.

\begin{proposition}[Strong functional correctness]
For every nonsingular square integer matrix $A$, the pure function and the instrumented execution return the same five data matrices.  They satisfy \eqref{eq:forward}--\eqref{eq:vinv}; $S$ satisfies the Smith predicate and equals the canonical reference Smith matrix.
\end{proposition}

\section{Termination by binary pivot descent}
\label{sec:termination}

The outer recursion is structural in the dimension.  The inner argument is arithmetic.

\begin{lemma}[Binary-size bridge]
\label{lem:sizebridge}
If $0<a$ and $2a\le b$, then $\natsize(a)<\natsize(b)$.
\end{lemma}

\begin{proof}
For $a>0$,
\[
  2^{\natsize(a)-1}\le a<2^{\natsize(a)}.
\]
Multiplying the left inequality by $2$ gives $2^{\natsize(a)}\le2a\le b$.  The characterization of \code{Nat.size} by powers of two therefore yields $\natsize(a)<\natsize(b)$.
\end{proof}

Let $p$ be the current pivot, $e$ a first-column entry with $p\nmid e$, and $q$ the pivot produced by the next pass.  The pass shape proves
\[
  q\mid p,\qquad q\mid e.
\]
If $|q|=|p|$, then $p\mid q$ and hence $p\mid e$, a contradiction.  Writing $|p|=k|q|$, nonzeroness and strictness exclude $k=0,1$, so
\begin{equation}
  2|q|\le |p|.
  \label{eq:halving}
\end{equation}
Since $q\ne0$, Lemma~\ref{lem:sizebridge} gives
\begin{equation}
  \natsize(|q|)<\natsize(|p|).
  \label{eq:sizedescent}
\end{equation}

The injection branch uses the same lemma.  Clearing and injection preserve $p$, and the chosen lower-right witness becomes a first-column entry not divisible by $p$.  The subsequent pass thus satisfies \eqref{eq:sizedescent}.

The actual termination measure is
\[
  \mu(A)=\mcode{Nat.size}(|A_{00}|).
\]
Every recursive stabilization call carries a proof that $\mu$ decreases.  The same induction gives
\[
  \mcode{passes}(\mcode{stabilizeFrom}(A))
    \le \mcode{Nat.size}(|A_{00}|),
\]
and, including the initial pass, a stabilization bound in terms of the input matrix width.  Totality and the multiplicative ``number of descent levels'' factor in the cost recurrence therefore use the same discrete quantity.

\section{Execution-aligned arithmetic traces}
\label{sec:execution}

\subsection{Primitive values and exact costs}

The arithmetic primitives operate on canonical sign-magnitude integers.  For operands $x,y$, let $\ell_x=\bitlen(x)$ and $\ell_y=\bitlen(y)$.  \Cref{tab:primitives} lists each charge form and its proved closed bound.  The ``exact cost'' column names the recursive counter stored in the primitive run.  The final theorem folds that counter.

\begin{table}[ht]
\centering
\footnotesize
\begin{tabularx}{\linewidth}{@{}p{0.17\linewidth}p{0.27\linewidth}X@{}}
\toprule
Leaf & Exact stored cost & Verified upper bound \\
\midrule
zero test & $1$ & $1$ \\
magnitude comparison & structural constructor count $C_{\mathrm{cmp}}(x,y)$ & $\ell_x+\ell_y+1$ \\
addition & \code{signMagnitudeAddSteps} & $1+2(\ell_x+\ell_y+1)^2$ \\
multiplication & \code{signMagnitudeMulSteps} & $1+\ell_y\bigl(1+(2\ell_x+\ell_y+2)^2\bigr)$ \\
quotient/remainder & \code{signMagnitudeDivModSteps} & $8+\ell_x(3+2\ell_y)+3\ell_x+3\ell_y$ \\
normalization unit & $1$ & $1$ \\
bounded Bezout block & one bounded-XGCD run, one gcd zero test, and, if $g\ne0$, two quotient/remainder runs & path-sensitive sum of those stored runs \\
\bottomrule
\end{tabularx}
\caption{Designated arithmetic leaves.  The arguments of the division row are numerator $x$ and divisor $y$.}
\label{tab:primitives}
\end{table}

If $X(x,y)$ is the stored bounded-XGCD cost and $Q(a,b)$ is the stored quotient/remainder cost, the opaque block has cost
\[
 B_{\mathrm{Bez}}(x,y)=
 \begin{cases}
   X(x,y)+1,&g=0,\\
   X(x,y)+1+Q(x,g)+Q(y,g),&g\ne0.
 \end{cases}
\]
Its value contains $g$, both Bezout coefficients, both exact quotients, and mutually inverse $2\times2$ matrices.  The branch comparison inside bounded XGCD is itself a costed structural magnitude comparison.

\subsection{The zero-cost structural boundary}
\label{sec:costboundary}

The theorem concerns the arithmetic grammar in \cref{tab:primitives}.  The following operations lie outside that grammar:

\begin{itemize}
  \item interpreting an integer coefficient in the sign-magnitude model is not charged as a separate conversion;
  \item matrix indexing, finite-index construction, transposition, reindexing, list concatenation, and construction of records or proof fields are structural and carry zero arithmetic cost;
  \item determinant evaluation, row or column updates, and dense products have no aggregate parent charge: they emit the addition and multiplication leaves they execute;
  \item the bounded Bezout block is represented by one opaque macro charge whose stored cost includes its internal XGCD, gcd test, and optional exact divisions; those internals are not emitted as separate child charges;
  \item decoding, serialization, storage allocation, garbage collection, compilation, and wall-clock execution lie outside the theorem.
\end{itemize}

The flat trace is assembled from the charge lists carried by the executed transitions and recursive children.  Matrix traversal and runtime representation remain outside this resource semantics.

\subsection{Shared runs, trace construction, and coverage}

Each leaf stores its operands, its primitive run, and a proof that the run is exactly the named primitive on those operands.  The algorithm uses the stored run's value; the trace fold uses the same run's cost.

Principal HNF and bounded-column HNF share a state transition kernel.  A Bezout transition binds one bounded-Bezout run.  Its forward matrix, inverse matrix, and opaque charge all refer to that run.  A reduction-above transition similarly binds one quotient/remainder run.  Three instrumented dense products then compute $B'$, $U'$, and $U'^{-1}$, and later control flow reads the resulting $B'$.

A coverage derivation connects the trace to the exact transition sequence: every transition refines the corresponding pure row step, the trace contains every transition charge exactly once, and each Bezout transformation and charge share the same stored run.

Smith execution composes covered HNF executions with costed divisibility searches, clearing, injection, certificate composition, and lower-right recursion.  A certificate composition executes exactly four dense products.

\begin{proposition}[Trace alignment]
For every nonsingular input, the execution value equals the pure Kannan--Bachem value.  The complete trace is assembled by concatenating the recorded charge lists of the covered transitions and recursive calls.  For each Bezout transition, the transformation matrices and the block charge refer to the same stored block run.  The reported cost is definitionally the fold of that trace.
\end{proposition}

\section{Coefficient bounds and explicit work recurrences}
\label{sec:bounds}

This section gives the dependency spine used in the proof.  Appendix~\ref{app:ledger} expands every auxiliary function appearing below, so the degree ledger can be recalculated directly from the displayed formulas.  We use $n$ for a matrix dimension, $d$ for the current active dimension, $k=d-1$, and $b,w$ for coefficient-width bounds; sans-serif symbols denote cost functions.

For a finite integer matrix $A$, write
\begin{equation}
 H(A)=\max_{i,j}|A_{ij}|,
 \qquad
 \mcode{matrixBitLength}(A)=\natsize\bigl(H(A)\bigr),
 \label{eq:matrixheightwidth}
\end{equation}
where the maximum is $0$ when the matrix has no entries.

\Cref{tab:boundspine} records the main domination theorems connecting actual trace folds to the recurrences displayed in this section.  Their width and nonsingularity assumptions are supplied by the preceding execution stage.

\begin{table}[ht]
\centering
\scriptsize
\begin{tabularx}{\linewidth}{@{}p{0.24\linewidth}p{0.31\linewidth}X@{}}
\toprule
Stage & Bounded execution quantity & Paper bound \\
\midrule
principal HNF & principal trace & $\mathsf{PHNF}$, \eqref{eq:phnf} \\
prepared HNF & prepared-principal trace & $\mathsf{Prepared}$, \eqref{eq:preparedprincipal} \\
bounded column & bounded-column trace & $\mathsf{Column}$, \eqref{eq:columnbound} \\
Smith pass & one pass trace & $\mathsf{Pass}$, \eqref{eq:passbound} \\
stabilization & full stabilization trace & $\mathsf{Stab}$, \eqref{eq:stabbound} \\
outer Smith recursion & complete Smith trace & $\mathsf{Smith}$, \eqref{eq:smithrec} \\
\bottomrule
\end{tabularx}
\caption{The theorem spine from instrumented executions to explicit work recurrences.}
\label{tab:boundspine}
\end{table}

\subsection{Scalar, determinant, and dense-product bounds}

For coefficient widths $x,y$, write
\begin{align}
 \AddCost(x,y)&=1+2(x+y+1)^2,\label{eq:addbound}\\
 \MulCost(x,y)&=1+y\bigl(1+(2x+y+2)^2\bigr),\label{eq:mulbound}\\
 \DivCost(x,y)&=8+x(3+2y)+3x+3y.\label{eq:divbound}
\end{align}
These are the addition, multiplication, and quotient/remainder bounds from \cref{tab:primitives}.  A length-$n$ dot-product term and a dense square product are bounded by
\begin{align}
 \mathsf T(n,l,r)
  &=\MulCost(l,r)+
    \AddCost\bigl(l+r,n(l+r+1)\bigr),\label{eq:dotbound}\\
 \mathsf{MatMul}(n,l,r)
  &=n^3\mathsf T(n,l,r).\label{eq:matmulbound}
\end{align}
Four-product certificate composition is
\begin{equation}
 \mathsf{Comp}(n,l,r)=
  2\mathsf{MatMul}(n,l,r)+2\mathsf{MatMul}(n,r,l).
 \label{eq:compbound}
\end{equation}
The asymmetry preserves the asymmetric multiplication counter.

For an $n\times n$ matrix of entry width $b$, the determinant-width bound is
\begin{equation}
  D(n,b)=n\bigl(\natsize(n)+b\bigr)+2.
  \label{eq:detbound}
\end{equation}
It follows from $|\det A|\le n!H(A)^n$ and the verified size lemmas for factorials, products, and powers.  The cost of one Bird determinant evaluation is the closed function $\mathsf{Det}(n,b)$ in \eqref{eq:detexecapp}.  Preparation is the recurrence
\begin{align}
 \mathsf{Prep}(0,b)&=0,\label{eq:prep0}\\
 \mathsf{Prep}(k+1,b)&=(k+1)\bigl(\mathsf{Det}(k,b)+1\bigr)
                       +\mathsf{Prep}(k,b).
 \label{eq:preprec}
\end{align}
The extra one is the charged zero test after each candidate determinant.

\subsection{Hermite widths and work}

The ready principal-HNF proof supplies closed widths
\[
 P(n,b),\qquad U_P(n,b),\qquad I_P(n,b)
\]
for the current matrix, forward multiplier, and inverse multiplier.  Their exact definitions are \eqref{eq:principalwidths}--\eqref{eq:principalinverseapp}.  Let
\begin{equation}
 \rho(n,f,i)=\natsize(n)+f+i
 \label{eq:transformwidth}
\end{equation}
and
\begin{equation}
\begin{split}
 \mathsf{Dense}(n,s,b,f,i)=s\bigl(&\mathsf{MatMul}(n,\rho(n,f,i),b)\\
  &+2\mathsf{MatMul}(n,\rho(n,f,i),\rho(n,f,i))\bigr).
\end{split}
\label{eq:densereplay}
\end{equation}
Here $s$ is the transition count.  The scalar-transition bound $\mathsf{Scalar}(w)$ is expanded in \eqref{eq:scalarapp}; it includes one bounded-XGCD block, its comparison control, external exact divisions, and the alternatives used by reduction-above and normalization.

The principal execution bound is exactly
\begin{equation}
 \mathsf{PHNF}(n,b)=n^3\mathsf{Scalar}(P(n,b))+
  \mathsf{Dense}\bigl(n,n^3,P(n,b),U_P(n,b),I_P(n,b)\bigr).
 \label{eq:phnf}
\end{equation}
Prepared principal HNF adds preparation and the two products that transport its certificate:
\begin{equation}
\begin{split}
 \mathsf{Prepared}(n,b)={}&\mathsf{Prep}(n,b)+\mathsf{PHNF}(n,b)\\
 &+\mathsf{MatMul}(n,U_P(n,b),1)
  +\mathsf{MatMul}(n,1,I_P(n,b)).
\end{split}
\label{eq:preparedprincipal}
\end{equation}

For bounded-column HNF, Appendix~\ref{app:ledger} defines the matrix and certificate widths
$B_C(n,b),U_C(n,b),I_C(n,b)$.  Its execution bound is
\begin{equation}
 \mathsf{Column}(n,b)=n\mathsf{Scalar}(b)+
   \mathsf{Dense}\bigl(n,n,B_C(n,b),U_C(n,b),I_C(n,b)\bigr).
 \label{eq:columnbound}
\end{equation}

For an active block of order $k+1$, $\delta$ denotes the determinant/minor width bound supplied to the prepared phase.  Let $t(k,b,\delta)$ be the maximum of the four phase-certificate widths in \eqref{eq:passtransformapp}, with the phase width rounded up to one.  In the Smith pass, $\delta$ is instantiated by the current determinant-width invariant.  The pass bound is
\begin{equation}
\begin{split}
 \mathsf{Pass}(k,b,\delta)={}&
   \mathsf{Column}(k+1,b)\\
 &+\mathsf{Prepared}\bigl(k+1,B_C(k+1,b)\bigr)\\
 &+\mathsf{Comp}\bigl(k+1,t(k,b,\delta),t(k,b,\delta)\bigr).
\end{split}
\label{eq:passbound}
\end{equation}

\subsection{Stabilization recurrence}

One divisibility decision has bound
\begin{equation}
 \mathsf{Search}(w)=2+\DivCost(w,w).
 \label{eq:searchbound}
\end{equation}
The two zero tests are separated by a quotient/remainder run when the divisor is nonzero.  Clearing and injection are
\begin{align}
 \mathsf{Clear}(d,w)&=(d-1)\DivCost(w,w)
       +\mathsf{MatMul}(d,w+1,w),\label{eq:clearbound}\\
 \mathsf{Inject}(d,w)&=\mathsf{MatMul}(d,w,1).
 \label{eq:injectbound}
\end{align}

Stabilization uses the following certificate-width quantities:
\begin{align}
 F(k,w)&=\max\bigl(t_0(k,w+1,w),w+1\bigr),\label{eq:stabfactor}\\
 W_{\mathrm{comp}}(k,w)&=\natsize(k+1)+F(k,w),\label{eq:stabcompositionwidth}\\
 R(d,w)&=(3w+2)W_{\mathrm{comp}}(d-1,w),\label{eq:stabtransformwidth}\\
 t_s(d,w)&=\natsize(d)+2R(d,w),\label{eq:stabexecwidth}\\
 \mathsf{Comp}_s(d,w)&=\mathsf{Comp}(d,t_s(d,w),t_s(d,w)).
 \label{eq:stabcompcost}
\end{align}
Here $t_0$ is the unrounded maximum in \eqref{eq:passtransformapp}; $F$ accounts for one pass certificate or the identity-sized continuation, $W_{\mathrm{comp}}$ for one certificate composition, and $R$ for the at-most $3w+2$ factor-composition contribution proved by pivot-size induction.

At dimension $d$, the below-pivot search examines at most $d-1$ entries and the lower-right search at most $(d-1)^2$.  The verified per-descent-level bound is
\begin{equation}
\begin{split}
 \mathsf{Level}(d,w)={}&
  (d-1)\mathsf{Search}(w)
  +(d-1)^2\mathsf{Search}(w)\\
 &+\mathsf{Pass}(d-1,w+1,w)
  +\mathsf{Clear}(d,w)+\mathsf{Inject}(d,w)\\
 &+3\mathsf{Comp}_s(d,w).
\end{split}
\label{eq:levelbound}
\end{equation}
The three stabilization branches contribute the trace fragments in \cref{tab:stabledger}.
\begin{table}[H]
\centering
\footnotesize
\begin{tabularx}{\linewidth}{@{}p{0.20\linewidth}X@{}}
\toprule
Branch & Charged execution fragments \\
\midrule
direct repeat & below-pivot search, one pass, recursive stabilization, one certificate composition \\
stable return & below-pivot search, clear, lower-block search \\
injection repeat & below-pivot search, clear, lower-block search, injection, one pass, recursive stabilization, three certificate compositions \\
\bottomrule
\end{tabularx}
\caption{Branch ledger for one stabilization level.  Equation~\eqref{eq:levelbound} takes a uniform bound over these fragments.}
\label{tab:stabledger}
\end{table}
The coefficient three in \eqref{eq:levelbound} comes from the injection branch.  Since each continuation strictly decreases a positive pivot size,
\begin{equation}
  \mathsf{Stab}(d,w)=(w+1)\mathsf{Level}(d,w).
  \label{eq:stabbound}
\end{equation}
The theorem uses the same strong induction on pivot size as the termination proof.

\subsection{Outer Smith recurrence and output width}

The recursively accumulated certificate width is
\begin{align}
 \mathsf{Tr}(0,w)&=0,\label{eq:transformrec0}\\
 \mathsf{Tr}(d+1,w)&=\natsize(d+1)+R(d+1,w)
       +\max\bigl(1,\mathsf{Tr}(d,w)\bigr).
 \label{eq:transformrec}
\end{align}
The actual arithmetic trace is bounded by
\begin{align}
 \mathsf{Smith}(0,w)&=0,\label{eq:smithrec0}\\
 \mathsf{Smith}(d+1,w)&=
  \mathsf{Stab}(d+1,w)+\mathsf{Smith}(d,w)\notag\\
 &\quad+\mathsf{Comp}\bigl(d+1,R(d+1,w),
                    \max(1,\mathsf{Tr}(d,w))\bigr).
 \label{eq:smithrec}
\end{align}
The lower-right recursion uses the same work width $w$: stabilized entries and the determinant of the lower block are bounded by the current determinant width.

For original input width $b$, set
\begin{equation}
 W(n,b)=\max\bigl(b,D(n,b)\bigr).
 \label{eq:workwidth}
\end{equation}
The two-parameter recursive work bound is
\begin{equation}
 \mathsf{Exec}(n,b)=\mathsf{Smith}(n,W(n,b)).
 \label{eq:closedexec}
\end{equation}
This is not yet the final one-variable monomial witness in \cref{thm:main}; that additional closure and codec substitution is performed in \cref{sec:poly}.

The transformation-certificate width and common output-entry width are
\begin{align}
 \mathsf{Cert}(n,b)&=\mathsf{Tr}(n,W(n,b)),\label{eq:certbound}\\
 b_{\mathrm{out}}(n,b)&=\max\bigl(D(n,b),\mathsf{Cert}(n,b)\bigr).
 \label{eq:outputentry}
\end{align}
Substituting $b_{\mathrm{out}}$ into \eqref{eq:outputencodingbound} gives the concrete five-matrix output work bound.

\section{Polynomial closure and the concrete witnesses}
\label{sec:poly}

The closure layer uses the base
\[
  \mathsf{base}(n,b)=n+b+2.
\]
A function $f:\N^2\to\N$ has a \emph{polynomial envelope} when the formal object stores a degree $d$ and a proof
\[
  f(n,b)\le\mathsf{base}(n,b)^d
\]
for all $n,b$.  The calculus includes constants, the two projections, addition, multiplication, maximum, fixed powers, truncated subtraction, \code{Nat.size}, monotone substitution, and additive dimension recursion.

The constructor degrees are
\begin{align}
 \deg(\text{constant }a)&=a, & \deg(n)&=\deg(b)=1,\notag\\
 \deg(f+g)&=\max(\deg f,\deg g)+1,
 &\deg(fg)&=\deg f+\deg g,\label{eq:degrules}\\
 \deg(\max(f,g))&=\max(\deg f,\deg g),
 &\deg(f^a)&=a\deg f,\notag\\
 \deg(\natsize\circ f)&=\max(\deg f,1)+1,\notag\\
 \deg\bigl(h(f,g)\bigr)&=(\max(\deg f,\deg g)+2)\deg h.\notag
\end{align}
An additive dimension recurrence receives both a base envelope and a step envelope and assigns two plus their maximum degree.  These rules are intentionally coarse: even the constant $a$ receives degree $a$.

Appendix~\ref{app:ledger} applies \eqref{eq:degrules} line by line to the recurrences in \cref{sec:bounds}.  The decisive part of the resulting ledger is
\begin{align*}
 \deg \mathsf{Tr}&=744, &
 \deg W&=5, &
 \deg \mathsf{Cert}&=(\max(1,5)+2)\cdot744=5{,}208,\\
 \deg \mathsf{Pass}&=307{,}232, &
 \deg \mathsf{Stab}&=307{,}238, &
 \deg \mathsf{Smith}&=307{,}241.
\end{align*}
The two final substitution steps are then
\begin{align}
 d_c&=(\max(1,5)+2)\cdot307{,}241
      =2{,}150{,}687,\label{eq:costdegreecalc}\\
 d_o&=(\max(1,5{,}208)+2)\cdot19
      =98{,}990.\label{eq:outputdegreecalc}
\end{align}
The factor $19$ is the envelope degree of the five-matrix output codec at a supplied entry width.

To replace the two-parameter base by the concrete input payload, the development proves
\[
 \mathsf{base}\bigl(n,\mcode{matrixBitLength}(A)\bigr)
   \le 2\bigl(\mcode{matrixBinarySize}(A)+1\bigr).
\]
Raising both sides to degree $d$ contributes the coefficient $2^d$.  Monotonicity from payload size to full encoding length gives \eqref{eq:maincost} and \eqref{eq:mainoutput}.

The analysis uses three levels of upper bound:
\begin{enumerate}
  \item the path-independent but still explicit recursive work functions such as $\mathsf{Exec}(n,b)$;
  \item their bivariate \code{PolyEnvelope} monomials in $n+b+2$;
  \item the final univariate witnesses $C_c(L+1)^{d_c}$ and $C_o(L+1)^{d_o}$ in the concrete encoding length $L$.
\end{enumerate}
The degree growth comes from specific closure steps.  Generic addition adds one degree at every use, uniform maxima discard branch information, and dense products use one worst-case width for all entries.  Dimension recursion applies one uniform per-level step, while monotone substitution multiplies outer and inner degrees.  The resulting witnesses certify polynomiality; their exponents are not sharp.

\section{Related work}
\label{sec:related}

\Cref{tab:related} compares the formal statements proved by the closest developments.

\begin{table}[ht]
\centering
\footnotesize
\renewcommand{\arraystretch}{1.08}
\begin{tabularx}{\linewidth}{@{}>{\raggedright\arraybackslash}p{0.15\linewidth}>{\raggedright\arraybackslash}p{0.20\linewidth}>{\raggedright\arraybackslash}p{0.17\linewidth}>{\raggedright\arraybackslash}p{0.22\linewidth}>{\raggedright\arraybackslash}X@{}}
\toprule
Work & Setting & Executable reduction & Certificate / canonicality & Complexity result \\
\midrule
Cano et al.~\cite{cano2016} & elementary divisor rings; constructive instances & verified algebraic algorithms & Smith form and uniqueness up to units & no concrete integer bit-cost theorem \\
Divason--Thiemann~\cite{divason2022} & arbitrary dimensions; Euclidean-domain instances & executable under computability assumptions & canonical Smith form and uniqueness & no efficient Kannan--Bachem bit-cost theorem \\
Bottesch et al.~\cite{bottesch2018} & integer LLL reduction & efficient integer implementation & reduced-basis correctness & formal polynomial running-time result \\
Ma et al.~\cite{ma2026} & reusable Lean decomposition schemas & instance-dependent certified proof paths & broad decomposition theorems, including SNF & no execution-aligned integer SNF cost theorem \\
This work & nonsingular square integer matrices & value-producing Kannan--Bachem execution & canonical $S$ and explicit $U,U^{-1},V,V^{-1}$ & fixed-polynomial binary arithmetic cost and output length \\
\bottomrule
\end{tabularx}
\caption{Comparison by formal theorem boundary.}
\label{tab:related}
\end{table}

Sergeraert analyzes extensions of the Kannan--Bachem organization to rectangular, rank-deficient matrices and identifies uncontrolled coefficient growth in a naive extension~\cite{sergeraert2024}.  The reported obstacle lies outside the original nonsingular square setting formalized here.

Cano et al. formalize constructive linear algebra over elementary divisor rings in Coq~\cite{cano2016}.  Divason and Thiemann give executable Smith algorithms for arbitrary matrix shapes in Isabelle/HOL under a small algebraic interface, instantiated in particular by Euclidean domains~\cite{divason2022}.  These developments cover broader algebraic settings.  They do not analyze the binary arithmetic cost of the integer Kannan--Bachem execution used here.

Bottesch, Haslbeck, and Thiemann provide the closest methodological precedent.  Their Isabelle/HOL development refines LLL to an efficient integer implementation and proves polynomial running-time bounds~\cite{bottesch2018}.  Time-credit and refinement frameworks offer more general accounts of verified complexity~\cite{zhan2018,gueneau2018}; our trace semantics is specialized to exact arithmetic and leaves structural cost unmodeled.

Ma, Wang, and Wen formalize reusable proof architectures for triangular, spectral, canonical, and Smith decompositions in Lean~\cite{ma2026}.  Their emphasis is theorem reuse across decomposition families.  Our result concerns one integer reduction and includes coefficient bounds, a binary codec, and an arithmetic-cost theorem for its execution.

\section{Limitations}
\label{sec:limitations}

\Cref{tab:resourcelayers} separates the proved arithmetic model from two stronger resource models.
\begin{table}[H]
\centering
\footnotesize
\begin{tabularx}{\linewidth}{@{}p{0.24\linewidth}p{0.20\linewidth}X@{}}
\toprule
Layer & Status & Content \\
\midrule
arithmetic trace & proved & exact stored costs of the designated sign-magnitude leaves, trace construction and coverage, and a fixed-polynomial bound \\
structural/RAM cost & not formalized & indexing, traversal, representation conversion, allocation, list and record construction, decoding, and serialization \\
compiled runtime & not formalized & compiler, garbage collector, and wall-clock execution \\
\bottomrule
\end{tabularx}
\caption{Resource-model boundary.}
\label{tab:resourcelayers}
\end{table}
A machine-time theorem would need a structural trace or an imperative refinement, followed by a simulation of the relevant data-structure operations in a stated machine model.

The trace uses two levels of granularity.  Addition, multiplication, and quotient/remainder are fine-grained leaves.  The bounded Bezout block is represented by one opaque charge for its bounded-XGCD computation and two optional exact divisions.  Comparing implementations would require expanding this block or proving a simulation theorem for its counter.

The implementation handles nonsingular square matrices.  It materializes both inverse transformations and uses dense products for certificate composition.  These choices simplify certificate checking, but they do not target high performance.

The polynomial witnesses are loose.  Smaller exponents would require path-sensitive aggregation, tighter coefficient estimates, and probably sparse or incremental certificate updates.

\section{Conclusion}

The verification has four components: the Smith computation, two-sided transformation certificates, termination, and arithmetic cost.  When divisibility fails, the next pass produces a proper common divisor, so \code{Nat.size} of the active pivot decreases.  This measure terminates stabilization and bounds its depth.  Coefficient theorems control the operands stored in the trace, and the work recurrences cover preparation, Hermite phases, searches, clearing, injection, and four-product certificate composition.  The verified codec turns the resulting dimension-and-width bounds into bounds on binary input and output lengths.

Thus one executable Kannan--Bachem reduction returns a canonical Smith matrix with explicit inverse certificates, while the trace generated by that execution satisfies a fixed polynomial arithmetic-cost bound.  Extending the theorem to singular or rectangular matrices, or refining the trace to a structural machine model, remains separate work.

\medskip
\noindent\textbf{AI Use Statement.} OpenAI Codex generated the Lean 4 proofs, and ChatGPT was used for implementation discussions.  The research ideas are the author's own; all statements, definitions, imports, and assumptions were manually reviewed, and the resulting proof terms were checked by Lean's kernel.

\bibliographystyle{plain}
\bibliography{references}

\appendix

\section{Expanded recurrence and degree ledger}
\label{app:ledger}

This appendix expands the abbreviations from \cref{sec:bounds}.  Each definition transcribes the corresponding Lean bound, and the degree column in \cref{tab:degreeledger} is obtained solely from \eqref{eq:degrules}.  The two final degrees can therefore be recalculated from the paper.

\subsection{Scalar and determinant ledger}

The bounded-XGCD coefficient widths and costs are
\begin{align}
 c_X(w)&=2(w+2),\\
 r_X(w)&=w+2c_X(w)+2,\\
 R_X(w)&=2w+2+\DivCost(r_X(w),w)
        +\MulCost(r_X(w)+1,w)\notag\\
 &\qquad+\AddCost(r_X(w),r_X(w)+1+w),\\
 E_X(w)&=1+\DivCost(w,w)+\MulCost(w+1,c_X(w))
       +\AddCost(c_X(w),w+1+c_X(w)),\\
 X(w)&=3+(2w+1)\bigl(E_X(w)+R_X(w)\bigr),\\
 \mathsf{Scalar}(w)&=X(w)+1+3\DivCost(w,w)+1.
 \label{eq:scalarapp}
\end{align}

For Bird evaluation, define
\begin{align*}
 I(n,b,s)&=s\bigl(\natsize(n)+b+3\bigr)+b+2,\\
 c&=I(n,b,n),\qquad
 \Delta=n(c+1),\qquad
 \Theta=n(c+b+1),\\
 E_{\det}(n,b,c)&=
   n\AddCost(c,\Delta)+1+\MulCost(\Delta,b)
   +n\mathsf T(n,c,b)\\
 &\qquad+\AddCost(\Delta+b,\Theta)+\MulCost(1,\Delta).
\end{align*}
Then
\begin{equation}
 \mathsf{Det}(n,b)=n^3E_{\det}\bigl(n,b,I(n,b,n)\bigr)
                  +\MulCost\bigl(1,I(n,b,n)\bigr).
 \label{eq:detexecapp}
\end{equation}

\subsection{Principal and bounded-column width ledger}

For a principal stage with target $t$ and incoming width $b$, let
\begin{align*}
 m(t,b)&=(t+1)((t+1)+(b+1))+2,\\
 o(t,b)&=(b+1)+m(t,b),\\
 c(t,b)&=2(o(t,b)+1)+1,\\
 r(t,b)&=2+c(t,b)+o(t,b),\\
 h(t,b)&=(b+1)+m(t,b)+r(t,b),\\
 i(t,b)&=(h(t,b)+1)+\bigl(t(r(t,b)+3)+1\bigr),\\
 P_{\mathrm{stage}}(t,b)&=(i(t,b)+1)+(m(t,b)+1),\\
 g(t,b)&=2b+2t+2t(t+b)+4.
\end{align*}
The common principal matrix width is
\begin{equation}
 P(n,b)=(b+1)+P_{\mathrm{stage}}(n,g(n,b)).
 \label{eq:principalwidths}
\end{equation}
The forward and inverse prefix widths are
\begin{align}
 U_P(0,b)&=0,\notag\\
 U_P(k+1,b)&=(k+1)+P(k+1,b)+k(k+b)+2,
 \label{eq:principalmultiplierapp}\\
 I_P(0,b)&=0,\notag\\
 I_P(k+1,b)&=(k+1)+b+k\bigl(k+P(k+1,b)\bigr)+2.
 \label{eq:principalinverseapp}
\end{align}

The bounded-column widths are
\begin{align}
 B_C(0,b)&=b+1,\notag\\
 B_C(k+1,b)&=b+k\bigl(\natsize(k+1)+2(b+1)+6\bigr)+1,
 \label{eq:bcintermediateapp}\\
 U_C(0,b)&=0,\notag\\
 U_C(k+1,b)&=\natsize(k+1)+B_C(k+1,b)+D(k,b),
 \label{eq:bcmultiplierapp}\\
 I_C(0,b)&=0,\notag\\
 I_C(k+1,b)&=\natsize(k+1)+b+D(k,B_C(k+1,b)).
 \label{eq:bcinverseapp}
\end{align}

\subsection{Pass certificate-width ledger}

For prepared principal HNF on an active block of order $k+1$, define
\begin{align*}
 P_M(k,b,\delta)&=\natsize(k+1)+\delta+D(k,b),\\
 P_I(k,b,\delta)&=\natsize(k+1)+b+D(k,\delta).
\end{align*}
With $b_L=B_C(k+1,b)$, the unrounded pass width is
\begin{equation}
\begin{split}
 t_0(k,b,\delta)=\max\{&U_C(k+1,b),I_C(k+1,b),\\
 &P_M(k,b_L,\delta),P_I(k,b_L,\delta)\},
\end{split}
\label{eq:passtransformapp}
\end{equation}
and the phase width used in composition is
$t(k,b,\delta)=\max(1,t_0(k,b,\delta))$.

\subsection{Degree calculation}

\begin{longtable}{@{}p{0.38\linewidth}r p{0.46\linewidth}@{}}
\caption{Degree ledger obtained from \eqref{eq:degrules}.  ``Subst.'' means one monotone substitution; ``dim-rec.'' means dimension recursion.}
\label{tab:degreeledger}\\
\toprule
Envelope / paper notation & Degree & Immediate construction \\
\midrule
\endfirsthead
\toprule
Envelope / paper notation & Degree & Immediate construction \\
\midrule
\endhead
addition $\AddCost$ & 9 & constants, two additions, square, multiplication \\
multiplication $\MulCost$ & 13 & operand product with a squared affine width \\
division $\DivCost$ & 11 & four affine/product summands \\
dot term $\mathsf T$ & 16 & $\MulCost+\AddCost$ \\
dense product $\mathsf{MatMul}$ & 19 & $n^3\mathsf T$ \\
certificate composition $\mathsf{Comp}$ & 22 & sum of two constant-scaled products \\
determinant width $D$ & 5 & $n(\natsize(n)+b)+2$ \\
determinant iteration width $I$ & 7 & affine sum containing $n(\natsize(n)+b+3)$ \\
determinant entry cost $E_{\det}$ & 34 & scalar/dot bounds at $I$ \\
determinant execution $\mathsf{Det}$ & 38 & $n^3E_{\det}+\MulCost$ \\
principal stage width $P_{\mathrm{stage}}$ & 19 & seven explicit stage-width constructions \\
principal stage boundary $g$ & 7 & quadratic dimension/input expression \\
principal matrix width $P$ & 172 & substitution $P_{\mathrm{stage}}(n,g(n,b))$ and addition \\
principal multiplier width $U_P$ & 175 & upper bound in \eqref{eq:principalmultiplierapp} \\
principal inverse width $I_P$ & 176 & upper bound in \eqref{eq:principalinverseapp} \\
bounded-column matrix width $B_C$ & 10 & upper bound in \eqref{eq:bcintermediateapp} \\
bounded-column multiplier $U_C$ & 16 & $\natsize+B_C+D$ \\
bounded-column inverse $I_C$ & 61 & $\natsize+b+D(n,B_C)$ \\
bounded-XGCD uniform cost $X$ & 39 & coefficient, reduction, Euclidean-step bounds \\
scalar transition $\mathsf{Scalar}$ & 42 & $X+1+3\DivCost(w,w)+1$ \\
principal dense replay & 551 & \eqref{eq:densereplay} at $P,U_P,I_P$ \\
principal execution $\mathsf{PHNF}$ & 7,312 & scalar substitution at $P$ plus dense replay \\
preparation $\mathsf{Prep}$ & 118 & dim-rec. with determinant-execution step \\
prepared principal $\mathsf{Prepared}$ & 7,315 & preparation, principal, two products \\
bounded-column execution $\mathsf{Column}$ & 205 & scalar transitions plus dense replay \\
clear / inject & 22 / 19 & \eqref{eq:clearbound} / \eqref{eq:injectbound} \\
pass transform $t_0$ & 244 & maximum of four certificate widths \\
pass execution $\mathsf{Pass}$ & 307,232 & column, prepared principal after subst., composition \\
stabilization factor / composition width & 244 / 245 & \eqref{eq:stabfactor} / \eqref{eq:stabcompositionwidth} \\
stabilization transform $R$ & 740 & product in \eqref{eq:stabtransformwidth} \\
stabilization execution width $t_s$ & 743 & $\natsize(d)+2R$ \\
stabilization composition cost $\mathsf{Comp}_s$ & 2,246 & composition at width $t_s$ \\
per-level cost $\mathsf{Level}$ & 307,236 & searches, pass, clear, inject, compositions \\
stabilization cost $\mathsf{Stab}$ & 307,238 & $(w+1)\mathsf{Level}$ \\
Smith transform $\mathsf{Tr}$ & 744 & dim-rec. with step degree 742 \\
work width $W$ & 5 & maximum of input and determinant widths \\
certificate width $\mathsf{Cert}$ & 5,208 & subst. $\mathsf{Tr}(n,W(n,b))$ \\
outer certificate composition & 6,709 & $\mathsf{Comp}(d,R,\max(1,\mathsf{Tr}))$ \\
Smith recursion step & 307,239 & stabilization plus outer composition \\
Smith work $\mathsf{Smith}$ & 307,241 & dim-rec. from zero and the step \\
final cost envelope $d_c$ & 2,150,687 & subst. $\mathsf{Smith}(n,W(n,b))$ \\
matrix encoding at entry width & 11 & equation \eqref{eq:matrixencodingbound} \\
five-matrix output encoding & 19 & equation \eqref{eq:outputencodingbound} \\
output entry width $b_{\mathrm{out}}$ & 5,208 & maximum of $D$ and $\mathsf{Cert}$ \\
final output envelope $d_o$ & 98,990 & output-codec subst. at $b_{\mathrm{out}}$ \\
\bottomrule
\end{longtable}

Two chains are worth checking explicitly.  First,
\[
 \deg\mathsf{Cert}=(\max(1,5)+2)\cdot744=5{,}208.
\]
Second, the pass degree $307{,}232$ dominates the stabilization level; multiplication by $w+1$ gives $307{,}238$, the outer composition has degree $6{,}709$, addition gives a step degree $307{,}239$, and dimension recursion gives
$307{,}241$.  Equations \eqref{eq:costdegreecalc} and
\eqref{eq:outputdegreecalc} then produce the final degrees.

\end{document}